\documentclass[10pt,twoside,english]{article}
\usepackage[latin9]{inputenc}
\usepackage[letterpaper]{geometry}
\usepackage{xcolor}
\usepackage{amsthm}
\usepackage{amsmath}
\usepackage{amssymb}
\usepackage{esint}
\usepackage[all]{xy}
\PassOptionsToPackage{normalem}{ulem}
\usepackage{ulem}

\makeatletter

\providecolor{lyxadded}{rgb}{0,0,1}
\providecolor{lyxdeleted}{rgb}{1,0,0}

\newcommand{\lyxaddress}[1]{
\par {\raggedright #1
\vspace{1.4em}
\noindent\par}
}
\theoremstyle{plain}
\newtheorem{thm}{\protect\theoremname}[section]
  \theoremstyle{remark}
  \newtheorem{rem}[thm]{\protect\remarkname}
  \theoremstyle{definition}
  \newtheorem{defn}[thm]{\protect\definitionname}
  \theoremstyle{plain}
  \newtheorem{lem}[thm]{\protect\lemmaname}

\usepackage[all]{xy}

\usepackage{multirow}
\usepackage{pslatex}
\newcommand{\m}{\mathsf{m}}
\newcommand{\dd}{\mathsf{d}}

\newcommand{\SP}{\mathsf{SP}}

\newcommand{\Case}{\mathsf{Case}}

\def\frontmatter@abstractheading{}

\date{}

\makeatother

\usepackage{babel}
  \providecommand{\definitionname}{Definition}
  \providecommand{\lemmaname}{Lemma}
  \providecommand{\remarkname}{Remark}
\providecommand{\theoremname}{Theorem}

\begin{document}

\title{Noncontextuality with Marginal Selectivity in Reconstructing Mental
Architectures}

\author{Ru Zhang\textsuperscript{*} and Ehtibar N. Dzhafarov\textsuperscript{**}}

\maketitle

\lyxaddress{\begin{center}
Department of Psychological Sciences, Purdue University\\\textsuperscript{*}zhang617@purdue.edu
\\\textsuperscript{**}ehtibar@purdue.edu (corresponding author)
\par\end{center}}
\begin{abstract}
We present a general theory of series-parallel mental architectures
with selectively influenced stochastically non-independent components.
A mental architecture is a hypothetical network of processes aimed
at performing a task, of which we only observe the overall time it
takes under variable parameters of the task. It is usually assumed
that the network contains several processes selectively influenced
by different experimental factors, and then the question is asked
as to how these processes are arranged within the network, e.g., whether
they are concurrent or sequential. One way of doing this is to consider
the distribution functions for the overall processing time and compute
certain linear combinations thereof (interaction contrasts). The theory
of selective influences in psychology can be viewed as a special application
of the interdisciplinary theory of (non)contextuality having its origins
and main applications in quantum theory. In particular, lack of contextuality
is equivalent to the existence of a \textquotedblleft hidden\textquotedblright{}
random entity of which all the random variables in play are functions.
Consequently, for any given value of this common random entity, the
processing times and their compositions (minima, maxima, or sums)
become deterministic quantities. These quantities, in turn, can be
treated as random variables with (shifted) Heaviside distribution
functions, for which one can easily compute various linear combinations
across different treatments, including interaction contrasts. This
mathematical fact leads to a simple method, more general than the
previously used ones, to investigate and characterize the interaction
contrast for different types of series-parallel architectures{\normalsize{}. }{\normalsize \par}


\markboth{Zhang and Dzhafarov}{Mental Architectures}

Running head: Mental Architectures

\newpage{}
\end{abstract}

\section{Introduction\label{sec:Introduction}}

The notion of a network of mental processes with components selectively
influenced by different experimental factors was introduced to psychology
in Saul Sternberg's (1969) influential paper. Sternberg considered
networks of processes $a,b,c,\ldots$ involved in performing a mental
task. Denoting their respective durations by $A,B,C\ldots$, the hypothesis
he considered was that the observed response time $T$ is $A+B+C+\ldots$
. One cannot test this hypothesis, Sternberg wrote, without assuming
that there are some factors, $\alpha,\beta,\gamma,\ldots$, that selectively
influence the durations $A,B,C\ldots$, respectively. Sternberg's
analysis was confined to stochastically independent $A,B,C,\ldots$,
and the consequences of the assumptions of seriality and selective
influences were tested on the level of the mean response times only. 

Subsequent development of these ideas was aimed at the entire distributions
of the response times and at a greater diversity and complexity of
mental architectures than just series of ``stages.'' This development
prominently includes Roberts and Sternberg (1993), Schweickert, Giorgini,
and Dzhafarov (2000), Schweickert and Townsend (1989), Townsend (1984,
1990a, 1990b), Townsend and Nozawa (1995), Townsend and Schweickert
(1989), and several other publications, primarily by James Townsend
and Richard Schweickert with colleagues. For an overview of these
developments see Dzhafarov (2003) and Schweickert, Fisher, and Sung
(2012). In the present context we should separately mention the development
of the ideas of \emph{stochastic ordering} of processing times in
Townsend (1984, 1990a) and Townsend and Schweickert (1989); as well
as the idea of \emph{marginal selectivity} (Townsend \& Schweickert,
1989).

The notion of selective influences also underwent a significant development,
having been generalized from stochastically independent random variables
to arbitrarily interdependent ones (Dzhafarov, 2003; Dzhafarov \&
Gluhovsky, 2006; Kujala \& Dzhafarov, 2008; Dzhafarov \& Kujala, 2010;
in press). The essence of the development is easy to understand using
two random variables (e.g., process durations) $A,B$ selectively
influenced by two respective factors $\alpha,\beta$. In Dzhafarov's
(2003) notation, this is written $(A,B)\looparrowleft(\alpha,\beta)$.
According to the definition given in Dzhafarov (2003), this means
that there are functions $f$ and $g$ and a random variable $R$
(a \emph{common source of randomness}) such that $f\left(\alpha,R\right)=A$
and $g\left(\beta,R\right)=B$. If such a choice of $\left(f,g,R\right)$
exists, it is not unique. For instance, $R$ can always be chosen
to have any distribution that is absolutely continuous with respect
to the usual Borel measure on the real line (e.g., a standard uniform,
or standard normal distribution, see Dzhafarov \& Gluhovsky, 2006).
However, a triple $\left(f,g,R\right)$ need not exist. It does not
exist, e.g., if marginal selectivity (Townsend \& Schweickert, 1989)
is violated, i.e., if the distribution of, say, $A$ at a given value
of $\alpha$ changes in response to changing $\beta$. But marginal
selectivity is not sufficient for the existence of a triple $\left(f,g,R\right)$.
Let, e.g., $\alpha$ and $\beta$ be binary factors, with values $1,2$
each, and let the correlation $\rho$ between $A$ and $B$ for a
treatment $\left(\alpha,\beta\right)$ be denoted $\rho_{\alpha\beta}$.
Then the triple in question does not exist if the correlations violated
the ``cosphericity test'' (Kujala \& Dzhafarov, 2008), also known
in quantum mechanics as Landau's inequality (Landau, 1989):
\begin{equation}
\left|\rho_{11}\rho_{12}-\rho_{21}\rho_{22}\right|\leq\overline{\rho}_{11}\overline{\rho}_{12}+\overline{\rho}_{21}\overline{\rho}_{22},
\end{equation}
where $\overline{\rho}_{\alpha\beta}=\sqrt{1-\rho_{\alpha\beta}^{2}}$.
There are many other known conditions that must be satisfied for the
existence of a triple $\left(f,g,R\right)$ when marginal selectivity
is satisfied (Dzhafarov \& Kujala, 2010, 2012a,b, 2013, 2014a). 

The allusion to quantum mechanics is not accidental: as shown in Dzhafarov
and Kujala (2012a,b), the theory of selective influences in psychology
can be viewed as a special application of the theory of (non)contextuality.
This theory is interdisciplinary (Dzhafarov \& Kujala, 2014b-d; Khrennikov,
2009), but its origins are in quantum theory, dating from Kochen and
Specker (1967) and John Bell's (1964, 1966) celebrated work. For the
modern state of the theory see Dzhafarov, Kujala, and Larsson (2015).
A simplified account of the (non)contextuality analysis of the example
given above is as follows. One labels each random variable in play
\emph{contextually}, i.e., by what property is being measured/recorded
under what treatment (context): 
\begin{equation}
\left(\underbrace{A_{\textnormal{value of }\alpha}}_{\textnormal{property: what is measured}}\right)^{\underbrace{\left(\textnormal{value of }\alpha,\textnormal{value of }\beta\right)}_{\textnormal{context: under what treatement}}},\left(\underbrace{B_{\textnormal{value of }\beta}}_{\textnormal{property: what is measured}}\right)^{\underbrace{\left(\textnormal{value of }\alpha,\textnormal{value of }\beta\right)}_{\textnormal{context: under what treatement}}}.
\end{equation}
The notation here is, of course, redundant, because the context and
property identifiers overlap, but we need now to emphasize the logic
rather than achieve notational convenience. Once the labeling is done,
one looks at all possible joint distributions imposable on all these
random variables, for all properties and all treatments. A system
is noncontextual if there exists such a joint distributions in which
any two random variables that represent the same property (``what
is measured'') are equal with probability 1. The latter is possible
only if the random variables representing the same property always
have the same distribution: in our case
\begin{equation}
\left(A_{\alpha}\right)^{\left(\alpha,\beta\right)}\sim\left(A_{\alpha}\right)^{\left(\alpha,\beta'\right)},\left(B_{\beta}\right)^{\left(\alpha,\beta\right)}\sim\left(B_{\beta}\right)^{\left(\alpha',\beta\right)}
\end{equation}
for any values $\alpha,\beta,\alpha',\beta'$ of the two factors.
This is called consistent connectedness (Dzhafarov, Kujala, \& Larsson,
2015), and in physics is known under a variety of names, including
(in certain paradigms) ``no-signaling condition'' (Cereceda, 2000;
Masanes, Acin, \& Gisin, 2006; Popescu \& Rohrlich, 1994). In psychology,
this is marginal selectivity. The definition of noncontextuality just
given is not the most general one, as the notion of contextuality
can be extended to inconsistently connected (violating marginal selectivity)
systems (Dzhafarov, Kujala, \& Larsson, 2015), but we do not need
this generality in this paper. What is important for us here is that
the existence of a joint distribution mentioned in our definition
is equivalent to the existence of a random variable $R$ and the functions
$f,g$ mentioned in the introductory paragraph. 

It is easy to show (Dzhafarov, 2003) that the existence of a triple
$\left(f,g,R\right)$ for given joint distributions of $\left(A,B\right)$
under different treatments $\left(\alpha,\beta\right)$ is equivalent
to the existence of a quintuple $\left(f',g',S,S_{A},S_{B}\right)$,
where $S,S_{A},S_{B}$ are random variables, such that $f'\left(\alpha,S,S_{A}\right)=A$
and $g'\left(\beta,S,S_{B}\right)=B$. In such a representation, one
can speak of a common source of randomness $S$ and specific sources
of randomness $S_{A},S_{B}$. In Dzhafarov, Schweickert, and Sung
(2004) this representation was used to investigate different series-parallel
arrangements of the hypothetical durations $A$ and $B$. The reason
this representation has been considered convenient is that if one
fixes the value $S=s$, then $f'\left(\alpha,s,S_{A}\right)=A_{c}$
and $g'\left(\beta,s,S_{B}\right)=B_{c}$ are stochastically independent
random variables. One can therefore use theorems proved for stochastically
independent selectively influenced components (Schweickert, Giorgini,
\& Dzhafarov, 2000) to obtain a general result by averaging across
possible values of $s$. For instance, let $\alpha,\beta$ be binary
factors (with values $1,2$ each), and let us assume that the observed
duration $T_{\alpha\beta}$ is $\min\left(A_{\alpha},B_{\beta}\right)$
for every treatment $\left(\alpha,\beta\right)$. Then $T_{\alpha\beta s}=\min\left(A_{\alpha s},B_{\beta s}\right)$
for every value $S=s$, and it is known that, for the independent
$A_{\alpha s},B_{\beta s}$ (satisfying a prolongation condition,
as explained below), 
\begin{equation}
\Pr\left(T_{11s}\leq t\right)-\Pr\left(T_{12s}\leq t\right)-\Pr\left(T_{21s}\leq t\right)+\Pr\left(T_{22s}\leq t\right)\leq0.
\end{equation}
Since this should be true for every value $S=s$, then it should also
be true that
\begin{equation}
C\left(t\right)=\Pr\left(T_{11}\leq t\right)-\Pr\left(T_{12}\leq t\right)-\Pr\left(T_{21}\leq t\right)+\Pr\left(T_{22}\leq t\right)\leq0.\label{eq:simple contrast}
\end{equation}
This follows from the fact that
\begin{equation}
\Pr\left(T_{\alpha\beta}\leq t\right)=\int\Pr\left(T_{\alpha\beta s}\leq t\right)\dd\m\left(s\right),
\end{equation}
where $\m\left(s\right)$ is the probability measure for $S$, and
the integration is over the space of all possible $s$. The linear
combination $C\left(t\right)$ in (\ref{eq:simple contrast}) is called
the \emph{interaction contrast of distributions functions}.

The Prolongation Assumption used in Dzhafarov et al. (2004), and derived
from Townsend (1984, 1990a) and Townsend and Schweickert (1989), is
that, for every $S=s$, 
\begin{equation}
\Pr\left(A_{1s}\leq t\right)\geq\Pr\left(A_{2s}\leq t\right),\;\Pr\left(B_{1s}\leq t\right)\geq\Pr\left(B_{2s}\leq t\right).\label{eq:prolongation simple}
\end{equation}
For this particular architecture, $T=\min\left(A,B\right)$, this
is the only assumption needed. To prove analogous results for more
complex mental architectures, however, one needs additional assumptions,
such as the existence of density functions for $A_{\alpha s},B_{\beta s}$
at every $s$, and even certain ordering of these density functions
in some vicinity $\left[0,\tau\right]$. 

The same results, however, can be obtained without these additional
assumptions, if one adopts the other, equivalent definition of selective
influences: $f\left(\alpha,R\right)=A$ and $g\left(\beta,R\right)=B$,
for some triple $\left(f,g,R\right)$. If such a representation exists,
then 
\begin{equation}
a_{\alpha r}=f\left(\alpha,r\right),b_{\beta r}=g\left(\beta,r\right)
\end{equation}
are deterministic quantities (real numbers), for every value $R=r$.
Any real number $x$ in turn can be viewed as a random variable whose
distribution function is a shifted Heaviside function
\begin{equation}
h\left(t-x\right)=\left\{ \begin{array}{c}
0,\:if\:t<x,\\
1,\:if\:t\geq x.
\end{array}\right.
\end{equation}
In particular, the quantity $t_{\alpha\beta r}=\min\left(a_{\alpha r},b_{\beta r}\right)$
for the simple architecture $T=\min\left(A,B\right)$ considered above
is distributed according to
\begin{equation}
h\left(t-t_{\alpha\beta r}\right)=h_{\alpha\beta r}\left(t\right).
\end{equation}
Let us see how inequality (\ref{eq:simple contrast}) can be derived
using these observations. 

We first formulate the \emph{(conditional) Prolongation Assumption},
a deterministic version of (\ref{eq:prolongation simple}): the assumption
is that $f,g,R$ can be so chosen that for every $R=r$, 
\begin{equation}
a_{1r}\leq a_{2r},\;b_{1r}\leq b_{2r}.\label{eq:prolongation 2x2}
\end{equation}
Without loss of generality, we can also assume, for any given $r$,
\begin{equation}
a_{1r}\leq b_{1r}\label{eq:default 2x2}
\end{equation}
(if not, rename $a$ into $b$ and vice versa). 
\begin{rem}
The prolongation Assumption clearly implies (\ref{eq:prolongation simple}).
Conversely, if (\ref{eq:prolongation simple}) holds, one can always
find functions $f,g,R$ for which the Prolongation Assumption holds
in the form above. For instance, one can choose $R=\left(S,S_{A},S_{B}\right)$,
take $S_{A}$ and $S_{B}$ to be uniformly distributed between 0 and
1, and choose $f\left(\alpha,\ldots\right),g\left(\beta,\ldots\right)$
to be the quantile functions for the hypothetical distributions of
$A$ and $B$ at the corresponding factor levels. 
\end{rem}
We next form the conditional (i.e., conditioned on $R=r$) interaction
contrast
\begin{equation}
c_{r}\left(t\right)=h_{11r}\left(t\right)-h_{12r}\left(t\right)-h_{21r}\left(t\right)+h_{22r}\left(t\right).\label{eq:HIC}
\end{equation}

\medskip{}

\noindent\textbf{Notation Convention.} When $r$ is fixed throughout
a discussion, we omit this argument and write $a_{\alpha},b_{\beta},t_{\alpha\beta},h_{\alpha\beta}(t),c\left(t\right)$
in place of $a_{\alpha r},b_{\beta r},t_{\alpha\beta r},h_{\alpha\beta r}(t),c_{r}(t)$.
(For binary factors $\alpha,\beta$, we also conveniently replace
$\alpha,\beta$ in indexation with $i,j$.)\medskip{}

Following this convention, there are three different arrangements
of $a_{1},a_{2},b_{1},b_{2}$ (for a given $R=r$) satisfying (\ref{eq:prolongation 2x2})-(\ref{eq:default 2x2}):

\begin{equation}
\begin{array}{cc}
\textnormal{(i)} & a_{1}\leq b_{1}\leq a_{2}\leq b_{2}\\
\textnormal{(ii)} & a_{1}\leq a_{2}\leq b_{1}\leq b_{2}\\
\textnormal{(iii)} & a_{1}\leq b_{1}\leq b_{2}\leq a_{2}
\end{array}\label{eq:3 arrangements}
\end{equation}
In all three cases, 

\begin{equation}
t_{11}=\min\left(a_{1},b_{1}\right)=a_{1}=\min\left(a_{1},b_{2}\right)=t_{12}.
\end{equation}
For arrangement (i) we have

\[
\xymatrix{\ar@{-}+<0ex,0ex>;[rrrrrr] & \overset{}{\underset{\begin{array}{c}
t_{11}=t_{12}=a_{1}\end{array}}{\bullet}} & \overset{\begin{array}{l}
+h_{11}\left(t\right)=1\\
-h_{12}\left(t\right)=-1\\
-h_{21}\left(t\right)=-0\\
+h_{22}\left(t\right)=0\\
=c\left(t\right)=0\\
\\
\end{array}}{\underset{\begin{array}{c}
\leq t<\end{array}}{}} & \overset{}{\underset{\begin{array}{c}
t_{21}=b_{1}\end{array}}{\bullet}} & \overset{\begin{array}{l}
+h_{11}\left(t\right)=1\\
-h_{12}\left(t\right)=-1\\
-h_{21}\left(t\right)=-1\\
+h_{22}\left(t\right)=0\\
=c\left(t\right)=-1\\
\\
\end{array}}{\underset{\begin{array}{c}
\leq t<\end{array}}{}} & \overset{}{\underset{\begin{array}{c}
t_{22}=a_{2}\end{array}}{\bullet}} & \,}
\]
This diagram shows the values of $h_{ijr}\left(t\right)$ and the
resulting values of $c_{r}\left(t\right)$ as $t$ changes with respect
to the fixed positions of $t_{ijr}$ (with index $r$ dropped everywhere).
Analogously, for arrangements (ii) and (iii), we have, respectively
\[
\xymatrix{\ar@{-}+<0ex,0ex>;[rrrr] & \overset{}{\underset{\begin{array}{c}
t_{11}=t_{12}=a_{1}\end{array}}{\bullet}} & \overset{\begin{array}{c}
c\left(t\right)=0\\
\\
\end{array}}{\underset{\begin{array}{c}
\leq t<\\
\\
\end{array}}{}} & \overset{}{\underset{\begin{array}{c}
t_{21}=t_{22}=a_{2}\end{array}}{\bullet}} & \,}
\]

and
\[
\xymatrix{\ar@{-}+<0ex,0ex>;[rrrrrr] & \overset{}{\underset{\begin{array}{c}
t_{11}=t_{12}=a_{1}\end{array}}{\bullet}} & \overset{\begin{array}{c}
c\left(t\right)=0\\
\\
\end{array}}{\underset{\begin{array}{c}
\leq t<\\
\\
\end{array}}{}} & \overset{}{\underset{\begin{array}{c}
t_{21}=b_{1}\end{array}}{\bullet}} & \overset{\begin{array}{c}
c\left(t\right)<0\\
\\
\end{array}}{\underset{\begin{array}{c}
\leq t<\\
\\
\end{array}}{}} & \overset{}{\underset{\begin{array}{c}
t_{22}=b_{2}\end{array}}{\bullet}} & \,}
\]
In all three cases, $c\left(t\right)$ is obviously zero for $t<t_{11}$
and $t\geq t_{22}$. We see that $c(t)=c_{r}\left(t\right)\leq0$
for all $t$ and every $R=r$. It follows that $C\left(t\right)\leq0$,
because
\begin{equation}
\Pr\left(T_{ij}\leq t\right)=\int h_{ijr}\left(t\right)\dd\mu\left(r\right),
\end{equation}
for $i,j\in\left\{ 1,2\right\} $, and
\begin{equation}
C\left(t\right)=\int c_{r}\left(t\right)\dd\mu\left(r\right)\leq0,
\end{equation}
where $\mu$ is the probability measure associated with $R$ and the
integration is over all possible $r$. We obtain the same result as
in (\ref{eq:simple contrast}), but in a very different way.

In this paper we extend this approach to other mental architectures
belonging to the class of series-parallel networks, those involving
other composition operations and possibly more than just two selectively
influenced processes. In doing so we follow a long trail of work mentioned
earlier. When dealing with multiple processes we follow Yang, Fific,
and Townsend (2013) in using high-order interaction contrasts. All
our results are replications or straightforward generalizations of
the results already known: the primary value of our work therefore
is not in characterizing mental architectures, but rather in demonstrating
a new theoretical approach and a new proof technique.

\subsection{Definitions, Terminology, and Notation}

Since we deal with the durations of processes rather than the processes
themselves, we use the term \emph{composition} to describe a function
that relates the durations of the components of a network to the overall
(observed) duration. Formally, a composition is a real-valued function
$t=t\left(a,b,\ldots,z\right)$ of an arbitrary number of real-valued
arguments. The arguments $a,b,\ldots,z$ are referred to as \emph{durations}
or \emph{components}. In this article, we will use $X\wedge Y\wedge\ldots\wedge Z$
to denote $\min(X,Y,\ldots,Z)$, and $X\vee Y\vee\ldots\vee Z$ to
denote $\max(X,Y\ldots,Z)$.

A\emph{ series-parallel composition} ($\SP$) is defined as follows.
\begin{defn}
\label{def:SP}(1) A single duration is an $\SP$ composition. (2)
If $X$ and $Y$ are $\SP$ compositions with disjoint sets of arguments,
then $X\wedge Y$, $X\vee Y$, and $X+Y$ are $\SP$ compositions.
(3) There are no other $\SP$ compositions than those construable
by Rules 1 and 2.\end{defn}
\begin{rem}
The requirement that $X$ and $Y$ in Rule 2 have disjoint sets of
arguments prevents expressions like $X\wedge X$ or $X+X\vee Y$.
But if the second $X$ in $X\wedge X$ is renamed into $X'$, or $X\vee Y$
in $X+X\vee Y$ is renamed into $Z$, then the resulting $X\wedge X'$
and $X+Z$ are legitimate $\SP$ compositions. This follows from the
generality of our treatment, in which different components of an $\SP$
composition may have arbitrary joint distributions: e.g., $X$ and
$X'$ in $X\wedge X'$ may very well be jointly distributed so that
$\Pr\left[X=X'\right]=1$. One should, however, always keep in minds
the pattern of selective influences: thus, if $X$ is influenced by
$\alpha$, then $Z$ is also influenced by $\alpha$ in $X+Z$ above. 
\end{rem}
Any $\SP$ composition is obtained by a successive application of
Rules 1 and 2 (the sequence being generally non-unique), and at any
intermediate stage of such a sequence we also have an $\SP$ composition
that we can term a \emph{subcomposition}.
\begin{defn}
\label{def:parallel-sequential}Two durations $X,Y$ in an $\SP$
composition are said to be \emph{parallel} or \emph{concurrent} if
there is a subcomposition of this $\SP$ composition of the form $\SP^{1}\left(X,X',\ldots\right)\wedge\SP^{2}\left(Y,Y',\ldots\right)$
(in which case $X,Y$ are said to be \emph{min-parallel}) or $\SP^{1}\left(X,X',\ldots\right)\vee\SP^{2}\left(Y,Y',\ldots\right)$
($X,Y$ are \emph{max-parallel}). $X,Y$ in an $\SP$ composition
are said to be \emph{sequential} or $serial$ if there is a subcomposition
of this $\SP$ composition of the form $\SP^{1}\left(X,X',\ldots\right)+\SP^{2}\left(Y,Y',\ldots\right)$.
\end{defn}

\begin{defn}
\label{def:homogeneous-mixed}An $\SP$ composition is called \emph{homogeneous}
if it does not contain both $\wedge$ and $\vee$ in it; if it does
not contain $\wedge$, it is denoted $\SP_{\vee}$; if it does not
contain $\vee$, it is denoted $\SP_{\wedge}$.
\end{defn}
The only $\SP$ composition that is both $\SP_{\wedge}$ and $\SP_{\vee}$
is a purely serial one: $a+b+\ldots+z$. Most of the results previously
obtained for mental networks are confined to homogeneous compositions.
We will not need this constraint for the most part.

Since we will be dealing with compositions of more than just two components,
we need to extend the definition of selective influences mentioned
above. In the formulation below, $\sim$ stands for ``has the same
distribution as.'' A treatment $\phi=\left(\lambda_{i_{1}}^{1},\ldots,\lambda_{i_{n}}^{n}\right)$
is a vector of values of the factors $\lambda^{1},\ldots,\lambda^{n}$,
the values of $\lambda^{k}$ ($k=1,\ldots,n$) being indicated by
subscripts, $\lambda_{i_{k}}^{k}$.
\begin{defn}
\label{def:SI}Random variables $(X^{1},\ldots,X^{n})$ are \emph{selectively
influenced} by factors $(\lambda^{1},\ldots,\lambda^{n})$, respectively,

\begin{equation}
(X^{1},\ldots,X^{n})\looparrowleft(\lambda^{1},\ldots,\lambda^{n}),
\end{equation}
if for some random variable $R$, whose distribution does not depend
on $(\lambda^{1},\ldots\lambda^{n})$, and for some functions $g_{1},\ldots,g_{n}$,

\begin{equation}
(X_{\phi}^{1},\ldots,X_{\phi}^{n})\sim(g_{1}(\lambda_{i_{1}}^{1},R),\ldots,g_{n}(\lambda_{i_{n}}^{n},R)),
\end{equation}
for any treatment $\phi=\left(\lambda_{i_{1}}^{1},\ldots,\lambda_{i_{n}}^{n}\right)$. 
\end{defn}
In the subsequent discussion we assume that all non-dummy factors
involved are binary in a completely crossed design (i.e., the overall
time $T$ is recorded for all $2^{n}$ vectors of values for $\phi$).
When we have random variables not influenced by any of these factors,
we will say they selectively influenced by an empty set of factors
(we could also, equivalently, introduce for them dummy factors, with
one value each).

\section{$\SP$ Compositions Containing Two Selectively Influenced Processes }

Consider two processes, with durations $A$ and $B$ in an $\SP$
composition. The overall duration of this $\SP$ composition can be
written as a function of $A,B$ and other components: $T=T(A,B,\ldots)$.
We assume that $A,B,$ and all other components are selectively influenced
by $\alpha$, $\beta$, and empty set, respectively: $(A,B,\ldots)\looparrowleft(\alpha,\beta,\emptyset)$.
Let each factor has two levels: $\alpha=1,2$ and $\beta=1,2$, with
four allowable treatments $(1,1)$, $(1,2)$, $(2,1)$, and $(2,2)$.
The corresponding overall durations (random variables) are written
as $T_{11},T_{12},T_{21},$ and $T_{22}$.

By Definition \ref{def:SI} of selective influences, each process
duration (a random variable) is a function of some random variable
$R$ and the corresponding factor: $A=a\left(\alpha,R\right)$, $B=b\left(\beta,R\right).$
For any given value $R=r$, the component durations are fixed numbers,
\begin{equation}
\begin{array}{ccc}
a\left(\alpha=1,r\right)=a_{1r}, &  & a\left(\alpha=2,r\right)=a_{2r},\\
b\left(\beta=1,r\right)=b_{1r}, &  & b\left(\beta=2,r\right)=b_{2r},\\
 & x\left(\emptyset,r\right)=x_{r},
\end{array}
\end{equation}
where $x$ is the value of any duration $X$ in the composition other
than $A$ and $B$. We assume that $R$ is chosen so that the Prolongation
Assumption (\ref{eq:prolongation 2x2}) holds, with the convention
(\ref{eq:default 2x2}). 

The overall duration $T$ at $R=r$ is also a fixed number, written
as (recall that we replace $\alpha,\beta$ in indexation with $i,j$)
\begin{equation}
T\left(a_{ir},\beta_{jr},\ldots\right)=t_{ijr},i,j\in\left\{ 1,2\right\} .
\end{equation}
The distribution function for $t_{ijr}$ is the shifted Heaviside
function $h_{ijr}\left(t\right)=h\left(t-t_{ijr}\right)$,
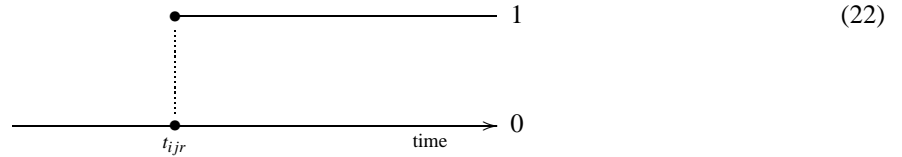
\begin{equation}
\xymatrix{ &  & \bullet\ar@{-}+<0ex,0ex>;[rrrr] &  &  &  & \;1\\
\ar@{-}[rrrrr]+<0ex,0ex>; &  & \underset{t_{ijr}}{{\bullet}}\ar@{.}[u] &  & \ar@{->}[rr]_{\textnormal{time}} &  & \;0
}
\end{equation}
The \emph{conditional interaction contrast} $c_{r}\left(t\right)$
is defined by (\ref{eq:HIC}). Denoting by $H_{ij}(t)$ the distribution
function of $T_{ij}$, we have 
\begin{equation}
H_{ij}\left(t\right)=\int_{\mathcal{R}}h_{ijr}\left(t\right)\dd\mu_{r},
\end{equation}
with $\mathcal{R}$ denoting the set of possible values of $R$. For
the observable (i.e., estimable from data) interaction contrast

\begin{equation}
C\left(t\right)=H_{11}\left(t\right)-H_{12}\left(t\right)-H_{21}\left(t\right)+H_{22}\left(t\right),\label{eq:C(t)}
\end{equation}
we have then

\begin{equation}
C\left(t\right)=\int_{\mathcal{R}}c_{r}\left(t\right)\dd\mu_{r}.\label{eq:basic integral}
\end{equation}
Note that it follows from our Prolongation Assumption that
\begin{equation}
H_{11}\left(t\right)\geq H_{12}\left(t\right),\;H_{21}\left(t\right)\geq H_{22}\left(t\right),\;H_{11}\left(t\right)\geq H_{21}\left(t\right),\;H_{12}\left(t\right)\geq H_{22}\left(t\right).
\end{equation}
We also define two conditional cumulative interaction contrasts (conditioned
on $R=r$):

\begin{equation}
c\left(0,t\right)=\int_{0}^{t}c\left(\tau\right)\dd\tau.\label{eq:cum2(0,t)}
\end{equation}
\begin{equation}
c\left(t,\infty\right)=\int_{t}^{\infty}c\left(\tau\right)\dd\tau=\lim_{u\rightarrow\infty}\int_{t}^{u}c\left(\tau\right)\dd\tau.\label{eq:cum2(t,infty)}
\end{equation}
The corresponding observable cumulative interaction contrasts are

\begin{equation}
C\left(0,t\right)=\int_{\mathcal{R}}c\left(0,t\right)\dd\mu_{r}=\int_{\mathcal{R}}\left(\int_{0}^{t}c\left(\tau\right)\dd\tau\right)\dd\mu_{r}=\int_{0}^{t}\left(\int_{\mathcal{R}}c\left(\tau\right)\dd\mu_{r}\right)\dd\tau=\int_{0}^{t}C\left(\tau\right)\dd\tau.\label{eq:Cum2(0,t)}
\end{equation}
\begin{equation}
C\left(t,\infty\right)=\int_{\mathcal{R}}c\left(t,\infty\right)\dd\mu_{r}=\int_{\mathcal{R}}\left(\int_{t}^{\infty}c\left(\tau\right)\dd\tau\right)\dd\mu_{r}=\int_{t}^{\infty}\left(\int_{\mathcal{R}}c\left(\tau\right)\dd\mu_{r}\right)\dd\tau=\int_{t}^{\infty}C\left(\tau\right)\dd\tau.\label{eq:Cum2(t,infty)}
\end{equation}
In these formulas we could switch the order of integration by Fubini's
theorem, because, for any interval of reals $I$, 

\begin{equation}
\int_{I\mathcal{\times R}}\left|c\left(\tau\right)\right|\dd\left(\tau\times\mu_{r}\right)\leq\int_{I\mathcal{\times R}}2\dd\left(\tau\times\mu_{r}\right)\leq2.
\end{equation}

\subsection{Four lemmas}

Recall the definition of $c_{r}\left(t\right)$ in (\ref{eq:HIC}).
We follow our Notation Convention and drop the index $r$ in $c_{r}\left(t\right)$
and all other expressions for a fixed $r$. 
\begin{lem}
\label{lem:In-any-SP}In any $\SP$ architecture, for any $r$, 
\[
t_{11}\leq t_{12}\wedge t_{21}\leq t_{12}\vee t_{21}\leq t_{22}.
\]
\end{lem}
\begin{proof}
Follows from the (nonstrict) monotonicity of the $\SP$ composition
in all arguments.\end{proof}
\begin{lem}
\label{lem:In-any-SP2}In any $\SP$ architecture, for any $r$, $c\left(t\right)$
equals 0 for all values of $t$ except for two cases: 

($\Case{}^{+}$) if $t_{11}\leq t<t_{12}\wedge t_{21}$, then $c\left(t\right)=1-0-0+0>0$,

and

($\Case^{-}$) if $t_{12}\vee t_{21}\leq t<t_{22}$, then $c\left(t\right)=1-1-1+0<0$.\end{lem}
\begin{proof}
By direct computation.\end{proof}
\begin{lem}
\label{lem:In-any-SP2-1}In any $\SP$ architecture, for any r, $c\left(t\right)\leq0$
for all values of t if and only if $t_{11}=t_{12}\wedge t_{21}$;
$c\left(t\right)\geq0$ for all values of t if and only if $t_{12}\vee t_{21}=t_{22}$. \end{lem}
\begin{proof}
Immediately follows from Lemma \ref{lem:In-any-SP2}.\end{proof}
\begin{lem}
\label{lem:In-any-SP3-1}In any $\SP$ architecture, for any $r$, 

(i) $c\left(0,t\right)=\int_{0}^{t}c\left(\tau\right)d\tau\geq0$
for any $t$ if and only if \textup{$-t_{11}+t_{12}+t_{21}-t_{22}\geq0$,
 }

(ii) $c\left(t,\infty\right)=\int_{t}^{\infty}c\left(\tau\right)d\tau\leq0$
for any $t$ if and only if \textup{$-t_{11}+t_{12}+t_{21}-t_{22}\leq0$,}

(iii) \textup{$\lim_{t\rightarrow\infty}c\left(0,t\right)=0$ }if
and only if \textup{$-t_{11}+t_{12}+t_{21}-t_{22}=0.$}

(iv) \textup{$\lim_{t\rightarrow0}c\left(t,\infty\right)=0$ }if and
only if \textup{$-t_{11}+t_{12}+t_{21}-t_{22}=0.$}\end{lem}
\begin{proof}
Without loss of generality, put $t_{12}\leq t_{21}$. We have
\[
c\left(0,t\right)=\begin{cases}
0 & \textnormal{if }t<t_{11}\\
\left(t-t_{11}\right) & \textnormal{if }t_{11}\leq t<t_{12}\\
\left(t-t_{11}\right)-\left(t-t_{12}\right)=t_{12}-t_{11} & \textnormal{if }t_{12}\leq t<t_{21}\\
\left(t-t_{11}\right)-\left(t-t_{12}\right)-\left(t-t_{21}\right)\\
\qquad\qquad\qquad=-t_{11}+t_{12}+t_{21}-t & \textnormal{if }t_{21}\leq t<t_{22}\\
\left(t-t_{11}\right)-\left(t-t_{12}\right)-\left(t-t_{21}\right)+\left(t-t_{22}\right)\\
\qquad\qquad\qquad=-t_{11}+t_{12}+t_{21}-t_{22} & \textnormal{if }t\geq t_{22}
\end{cases}
\]
The expressions for the first three cases are obviously nonnegative.
If $-t_{11}+t_{12}+t_{21}-t_{22}\geq0$, then $c\left(0,t\right)\geq0$
for all $t$ in the last case ($t\geq t_{22}$). With $-t_{11}+t_{12}+t_{21}-t_{22}\geq0$,
we have $-t_{11}+t_{12}+t_{21}-t\geq t_{22}-t\geq0$ for the fourth
case ($t_{21}\leq t<t_{22}$). Hence $c\left(0,t\right)\geq0$ for
all $t$ if $-t_{11}+t_{12}+t_{21}-t_{22}\geq0$. Conversely, if $c\left(0,t\right)\geq0$
for all $t$, then it is also true for $t\geq t_{22}$, whence $-t_{11}+t_{12}+t_{21}-t_{22}\geq0$.

The proof for $c\left(t,\infty\right)=\int_{t}^{\infty}c\left(\tau\right)d\tau$
requires replacing it first with $\int_{t}^{u}c\left(\tau\right)d\tau\leq0$
for some $u>t_{22}$. We have 
\[
\int_{t}^{u}c\left(\tau\right)d\tau=\begin{cases}
\left(u-t_{11}\right)-\left(u-t_{12}\right)-\left(u-t_{21}\right)+\left(u-t_{22}\right)\\
\qquad\qquad\qquad=-t_{11}+t_{12}+t_{21}-t_{22} & \textnormal{if }t<t_{11}\\
\left(u-t\right)-\left(u-t_{12}\right)-\left(u-t_{21}\right)+\left(u-t_{22}\right)\\
\qquad\qquad\qquad=-t+t_{12}+t_{21}-t_{22} & \textnormal{if }t_{11}\leq t<t_{12}\\
\left(u-t\right)-\left(u-t\right)-\left(u-t_{21}\right)+\left(u-t_{22}\right)\\
\qquad\qquad\qquad=t_{21}-t_{22} & \textnormal{if }t_{12}\leq t<t_{21}\\
\left(u-t\right)-\left(u-t\right)-\left(u-t\right)+\left(u-t_{22}\right)\\
\qquad\qquad\qquad=t-t_{22} & \textnormal{if }t_{21}\leq t<t_{22}\\
\left(u-t\right)-\left(u-t\right)-\left(u-t\right)+\left(u-t\right)\\
\qquad\qquad\qquad=0 & \textnormal{if }t\geq t_{22}
\end{cases}
\]
The expressions for the last three cases are obviously nonpositive.
If $-t_{11}+t_{12}+t_{21}-t_{22}\leq0$, then $\int_{t}^{u}c^{(2)}\left(\tau\right)d\tau\leq0$
for all $t$ in the first case ($t<t_{11}$). With $-t_{11}+t_{12}+t_{21}-t_{22}\leq0$,
we have $-t+t_{12}+t_{21}-t_{22}\leq t_{11}-t<0$ for the second case
($t_{11}\leq t<t_{12}$). Hence $\int_{t}^{u}c\left(\tau\right)d\tau\leq0$
for all $t$ if $-t_{11}+t_{12}+t_{21}-t_{22}\leq0$ . Since in all
expressions $u$ is algebraically eliminated, they remain unchanged
as $u\rightarrow\infty$. Conversely, if $c\left(t,\infty\right)\leq0$
for all $t$, then it is also true for $t<t_{11}$, whence $-t_{11}+t_{12}+t_{21}-t_{22}\leq0$.

The statements (iii) and (iv) follow trivially.
\end{proof}

\subsection{Parallel Processes}

\subsubsection{Simple Parallel Architectures of Size 2}

A simple parallel architecture corresponds to one of the two compositions:
$T=A\wedge B$ or $T=A\vee B$, with $(A,B)\looparrowleft(\alpha,\beta)$.
Recall the definition of $C\left(t\right)$ in (\ref{eq:C(t)}).
\begin{thm}
\label{thm:Basis} For $T=A\wedge B$, we have $c\left(t\right)\leq0$
for any $r,t$; consequently, $C\left(t\right)\leq0$ for any $t$.
For $T=A\vee B$, we have $c\left(t\right)\geq0$ for any $r,t$;
consequently, $C\left(t\right)\geq0$ for any $t$. \end{thm}
\begin{proof}
For $T=A\wedge B$ with the Prolongation Assumption (\ref{eq:prolongation 2x2})-(\ref{eq:default 2x2}),
we have 
\[
t_{11}=a_{1}\wedge b_{1}=a_{1},\:t_{12}=a_{1}\wedge b_{2},\:t_{21}=a_{2}\wedge b_{1}.
\]
It follows that 
\[
t_{12}\wedge t_{21}=a_{1}\wedge b_{2}\wedge a_{2}\wedge b_{1}=a_{1}=t_{11}.
\]
By Lemma \ref{lem:In-any-SP2-1}, $c\left(t\right)\leq0$. As $C\left(t\right)$
in (\ref{eq:basic integral}) preserves the sign of $c\left(t\right)$,
we have $C\left(t\right)\leq0$. For $T=A\wedge B,$ we have 
\[
t_{22}=a_{2}\vee b_{2},\:t_{12}=a_{1}\vee b_{2},\:t_{21}=a_{2}\vee b_{1}.
\]
It follows that 
\[
t_{12}\vee t_{21}=a_{1}\vee b_{2}\vee a_{2}\vee b_{1}=t_{22},
\]
whence, by Lemma \ref{lem:In-any-SP2-1}, $c\left(t\right)\geq0$
and therefore $C\left(t\right)\geq0$. 
\end{proof}

\subsubsection{Two Parallel Processes in an Arbitrary $\SP$ Network}

Consider now a composition $\SP(A,B,\ldots)$ with $(A,B,\ldots)\looparrowleft(\alpha,\beta,\emptyset)$. 
\begin{lem}
\label{lem:parallel decomposition 2}If $A,B$ in $\SP(A,B,\ldots)$
are parallel, then $\SP(A,B,\ldots)$ can be presented as $A'\wedge B'$
if they are min-parallel, or as $A'\vee B'$ if they are max-parallel,
so that $(A',B')\looparrowleft(\alpha,\beta)$ and, for any fixed
$R=r$, the Prolongation Assumption holds.\end{lem}
\begin{proof}
By Definitions \ref{def:SP} and \ref{def:parallel-sequential}, if
$A,B$ are min-parallel, then $\SP_{\wedge}(A,B,\ldots)$ can be presented
either as
\[
\SP^{1}(A,\ldots)\wedge\SP^{2}(B,\ldots)
\]
or

\[
\left(\SP^{1}(A,\ldots)\wedge\SP^{2}(B,\ldots)+X\right)\wedge Y,
\]
or else
\[
\left(\SP^{1}(A,\ldots)\wedge\SP^{2}(B,\ldots)\wedge X\right)+Y,
\]
where $B$ does not enter in $\SP^{1}$ and $A$ does not enter in
$\SP^{2}$. On renaming 
\[
\underset{=A'}{\underbrace{\SP^{1}(A,\ldots)}}\wedge\underset{=B'}{\underbrace{\SP^{2}(B,\ldots)}},
\]
\[
\left(\SP^{1}(A,\ldots)\wedge\SP^{2}(B,\ldots)+X\right)\wedge Y=\underset{=A'}{\underbrace{\left(\SP^{1}(A,\ldots)+X\right)}}\wedge\underset{=B'}{\underbrace{\left(\SP^{2}(B,\ldots)+X\right)\wedge Y}},
\]
and
\[
\left(\SP^{1}(A,\ldots)\wedge\SP^{2}(B,\ldots)\wedge X\right)+Y=\underset{=A'}{\underbrace{\left(\SP^{1}(A,\ldots)+Y\right)}}\wedge\underset{=B'}{\underbrace{\left(\SP^{2}(B,\ldots)\wedge X+Y\right)}},
\]
we have, obviously, $(A',B')\looparrowleft(\alpha,\beta)$. Fixing
$R=r$, by the (nonstrict) monotonicity of $\SP$ compositions, 
\[
a'_{1}=\SP{}^{1}(a_{1},\ldots)\leq\SP{}^{1}(a_{2},\ldots)=a'_{2}
\]
and 
\[
b'_{1}=\SP{}^{2}(b_{1},\ldots)\leq\SP{}^{2}(b_{2},\ldots)=b'_{2}
\]
We can also put $a'_{1}=\SP{}^{1}(a_{1},\ldots)$ $\leq$ $\SP{}^{2}(b_{1},\ldots)=b'_{1}$
(otherwise we can rename the variables). The proof for the max-parallel
case is analogous.\end{proof}
\begin{thm}
\label{thm:2x2 general}If $A,B$ in $\SP(A,B,\ldots)$ are min-parallel,
then $c\left(t\right)\leq0$ for any $r,t$; consequently, $C\left(t\right)\leq0$
for any $t$. If $A,B$ in $\SP(A,B,\ldots)$ are max-parallel, then
$c\left(t\right)\geq0$ for any $r,t$; consequently, $C\left(t\right)\geq0$
for any $t$.\end{thm}
\begin{proof}
Immediately follows from Lemma \ref{lem:parallel decomposition 2}
and Theorem \ref{thm:Basis}.
\end{proof}

\subsection{Sequential Processes}

\subsubsection{Simple Serial Architectures of Size 2}

Simple serial architectures of size 2 corresponds to the $\SP$ composition
$T=A+B$, with $(A,B)\looparrowleft(\alpha,\beta)$. Recall the definitions
of the two cumulative interaction contrasts: (\ref{eq:cum2(0,t)})-(\ref{eq:cum2(t,infty)})
and (\ref{eq:Cum2(0,t)})-(\ref{eq:Cum2(t,infty)}). 
\begin{thm}
\label{thm:simple2plus} If $T=A+B$, then $c\left(0,t\right)\geq0$
and $c\left(t,\infty\right)\leq0$ for any $r,t$; moreover, 
\[
\lim_{t\rightarrow\infty}c\left(0,t\right)=\lim_{t\rightarrow0}c\left(t,\infty\right)=0,
\]
 for any $r,t$. Consequently, $C\left(0,t\right)\geq0$, $C\left(t,\infty\right)\leq0$
for any $t$, and 
\[
\lim_{t\rightarrow\infty}C\left(0,t\right)=\lim_{t\rightarrow0}C\left(t,\infty\right)=0
\]
. \end{thm}
\begin{proof}
Follows immediately from Lemma \ref{lem:In-any-SP3-1}, since
\[
-t_{11}+t_{12}+t_{21}-t_{22}=-\left(a_{1}+b_{1}\right)+\left(a_{1}+b_{2}\right)+\left(a_{2}+b_{1}\right)-\left(a_{2}+b_{2}\right)=0.
\]

\end{proof}

\subsubsection{Two Sequential Processes in an Arbitrary $\SP$ Network}

Consider now a composition $\SP(A,B,\ldots)$ with $(A,B,\ldots)\looparrowleft(\alpha,\beta,\emptyset)$. 
\begin{thm}
\label{thm:2x2 general-1-1}If $A$ and $B$ are sequential in an
$\SP(A,B,\ldots)$ composition, then one or both of the following
statements hold: 

(i) $c\left(0,t\right)\geq0$ for any $r,t$, and $C\left(0,t\right)\geq0$
for any $t$, 

(ii) $c\left(t,\infty\right)\leq0$ for any $r,t$, and $C\left(t,\infty\right)\leq0$
for any $t$.\end{thm}
\begin{proof}
In accordance with Definitions \ref{def:SP} and \ref{def:parallel-sequential},
$\SP(A,B,\ldots)$ with sequential $A,B$ can be presented as either
\begin{equation}
\left(\SP^{1}(A,\ldots)+\SP^{2}(B,\ldots)\right)\wedge X+Y\label{eq:former}
\end{equation}
or
\begin{equation}
\left(\SP^{1}(A,\ldots)+\SP^{2}(B,\ldots)\right)\vee X+Y\label{eq:latter}
\end{equation}
(note that any $Z$ in $\SP^{1}(A,\ldots)+\SP^{2}(B,\ldots)+Z$ can
be absorbed by either of the first two summands). For both cases,
by the monotonicity of $\SP$ compositions, for any $R=r$, $\SP^{1}(a_{1},\ldots)\leq\SP^{1}(a_{2},\ldots)$,
$\SP^{2}(b_{1},\ldots)\leq\SP^{2}(b_{2},\ldots)$, and we can always
assume $\SP^{1}(a_{1},\ldots)\leq\SP^{2}(b_{1},\ldots)$. Denoting
the durations of $\SP^{1}(a_{i},\ldots)+\SP^{2}(b_{j},\ldots)$ by
$t'_{ij}$, we have therefore, by Theorem \ref{thm:simple2plus},
$-t'_{11}+t'_{12}+t'_{21}-t'_{22}=0$. Denoting the durations of $X$
and $Y$ by $t'$ and $t''$, respectively, in the case (\ref{eq:former})
we have
\[
t_{ij}=t'_{ij}\wedge t'+t''.
\]
By Lemma \ref{lem:In-any-SP3-1}, all we have to show is that $-t{}_{11}+t{}_{12}+t{}_{21}-t{}_{22}\geq0$.
It is easy to see that $t''$ does not affect this linear combination,
and its value is (assuming $t'_{12}\leq t'_{21}$, without loss of
generality) 
\[
\begin{cases}
0 & \textnormal{if }t'<t'_{11}\\
-t'_{11}+t' & \textnormal{if }t'_{11}\leq t'<t'_{12}\\
-t'_{11}+t'_{12} & \textnormal{if }t'_{12}\leq t'<t'_{21}\\
-t'_{11}+t'_{12}+t'_{21}-t' & \textnormal{if }t'_{21}\leq t'<t'_{22}\\
-t'_{11}+t'_{12}+t'_{21}-t'_{22} & \textnormal{if }t'\geq t'_{22}.
\end{cases}
\]
The nonnegativity of the first three expressions is obvious, the fifth
one is zero, and the forth expression is larger than the fifth because
$t'<t_{22}'$. 

The proof for the case (\ref{eq:latter}) is analogous.
\end{proof}
If the $\SP$ composition with sequential $A,B$ is homogeneous (Definition
\ref{def:homogeneous-mixed}), the statement of theorem can be made
more specific. 
\begin{thm}
If $A$ and $B$ are sequential in an $\SP_{\wedge}(A,B,\ldots)$
composition, then $c\left(0,t\right)\geq0$ for any $r,t$, and $C\left(0,t\right)\geq0$
for any $t$; if the composition is $\SP_{\vee}(A,B,\ldots)$, then
$c\left(t,\infty\right)\leq0$ for any $r,t$, and $C\left(t,\infty\right)\leq0$
for any $t$.
\end{thm}

\section{Multiple Processes}

We now turn to networks containing $n\geq2$ processes with durations
$(X^{1},\ldots,X^{n})$, selectively influenced by factors $(\lambda^{1},\ldots,\lambda^{n})$.
In other words, we deal with compositions $\SP(X^{1},\ldots,X^{n},\ldots)$
such that $(X^{1},\dots,X^{n},\ldots)\looparrowleft(\lambda^{1},\ldots,\lambda^{n},\emptyset)$,
where each $\lambda^{k}$ is binary, with values 1,2. There are $2^{n}$
allowable treatments and $2^{n}$ corresponding overall durations,
$T_{11\ldots1},T_{11\ldots2},\ldots,T_{22\ldots2}$. According to
Definition \ref{def:SI} of selective influences, each process duration
here is a function of some random variable $R$ and of the corresponding
factor, $X^{k}=x^{k}(R,\lambda^{k})$. For any fixed value $R=r$,
these durations are fixed numbers for any given treatment, and so
is the overall, observed value of the $\SP$ composition. We denote
them

\begin{equation}
x^{k}(r,\lambda^{k}=1)=x_{1r}^{k},\:x^{k}(r,\lambda^{k}=2)=x_{2r}^{k},\label{eq:function for n}
\end{equation}
and

\begin{equation}
T(x_{i_{1}r}^{1},x_{i_{2}r}^{2},\ldots,x_{i_{n}r}^{n},\ldots),\ldots=t_{i_{1}i_{2}...i_{n}r},
\end{equation}
where $i_{1},i_{2},\ldots,i_{n}\in\left\{ 1,2\right\} .$ The distribution
function for $t_{i_{1}i_{2}...i_{n}r}$ is a shifted Heaviside function
\begin{equation}
h_{i_{1}i_{2}...i_{n}r}\left(t\right)=\left\{ \begin{array}{c}
0,\:if\:t<t_{i_{1}i_{2}...i_{n}r}\\
1,\:if\:t\geq t_{i_{1}i_{2}...i_{n}r}
\end{array}.\right.
\end{equation}
Denoting by $H_{i_{1}i_{2}...i_{n}}\left(t\right)$ the distribution
function of $T_{i_{1}i_{2}\ldots i_{n}}$, we have
\begin{equation}
H_{i_{1}i_{2}...i_{n}}\left(t\right)=\int_{\mathcal{R}}h_{i_{1}i_{2}...i_{n}r}\left(t\right)\dd\mu_{r}.
\end{equation}

Conditioned on $R=r$, the $n$-\emph{th order interaction contrast}
is defined in terms of mixed finite differences as

\begin{equation}
c_{r}^{\left(n\right)}\left(t\right)=\Delta_{i_{1}}\Delta_{i_{2}}\ldots\Delta_{i_{n}}h_{i_{1}i_{2}...i_{n}r}\left(t\right),
\end{equation}
which, with some algebra can be shown to be equal to
\begin{equation}
c_{r}^{\left(n\right)}\left(t\right)=\underset{i_{1},i_{2},\ldots,i_{n}}{\sum}\left(-1\right)^{n+\sum_{k=1}^{n}i_{k}}h_{i_{1}\ldots i_{n}r}\left(t\right).
\end{equation}
Thus,
\begin{equation}
c_{r}^{\left(1\right)}\left(t\right)=\Delta_{i_{1}}h_{i_{1}r}\left(t\right)=h_{1r}\left(t\right)-h_{2r}\left(t\right)=\underset{i_{1}}{\sum}\left(-1\right)^{1+i_{1}}h_{i_{1}r}\left(t\right),
\end{equation}
\begin{equation}
\begin{array}{l}
c_{r}^{\left(2\right)}\left(t\right)=\Delta_{i_{1}}\Delta_{i_{2}}h_{i_{1}i_{2}r}\left(t\right)=\left[h_{11r}\left(t\right)-h_{12r}\left(t\right)\right]-\left[h_{21r}\left(t\right)-h_{22r}\left(t\right)\right]\\
=h_{11r}\left(t\right)-h_{12r}\left(t\right)-h_{21r}\left(t\right)+h_{22r}\left(t\right)=\underset{i_{1},i_{2}}{\sum}\left(-1\right)^{2+i_{1}+i_{2}}h_{i_{1}i_{2}r}\left(t\right),
\end{array}
\end{equation}
\begin{equation}
\begin{array}{l}
c_{r}^{\left(3\right)}\left(t\right)=\Delta_{i_{1}}\Delta_{i_{2}}\Delta_{i_{3}}h_{i_{1}i_{2}i_{3}r}\left(t\right)\\
=\left\{ \left[h_{111r}\left(t\right)-h_{112r}\left(t\right)\right]-\left[h_{121r}\left(t\right)-h_{122r}\left(t\right)\right]\right\} -\left\{ \left[h_{211r}\left(t\right)-h_{212r}\left(t\right)\right]-\left[h_{221r}\left(t\right)-h_{222r}\left(t\right)\right]\right\} \\
=h_{111r}\left(t\right)-h_{112r}\left(t\right)-h_{121r}\left(t\right)-h_{211r}\left(t\right)+h_{122r}\left(t\right)+h_{212r}\left(t\right)+h_{221r}\left(t\right)-h_{222r}\left(t\right)\\
=\underset{i_{1},i_{2},i_{3}}{\sum}\left(-1\right)^{3+i_{1}+i_{2}+i_{3}}h_{i_{1}i_{2}i_{3}r}\left(t\right),
\end{array}
\end{equation}
etc. The observable distribution function interaction contrast of
order $n$ is defined as 
\begin{equation}
C^{(n)}\left(t\right)=\int_{\mathcal{R}}c_{r}^{(n)}\left(t\right)\dd\mu_{r}.
\end{equation}
By straightforward calculus this can be written in extenso as 
\begin{equation}
C^{\left(n\right)}\left(t\right)=\underset{i_{1},i_{2},\ldots,i_{n}}{\sum}\left(-1\right)^{n+\sum_{k=1}^{n}i_{k}}H_{i_{1}\ldots i_{n}}\left(t\right),
\end{equation}
or, in terms of finite differences,

\begin{equation}
C^{\left(n\right)}\left(t\right)=\Delta_{i_{1}}\Delta_{i_{2}}\ldots\Delta_{i_{n}}H_{i_{1}i_{2}...i_{n}}\left(t\right).
\end{equation}
This is essentially the high-order interaction contrast used by Yang,
Fific, and Townsend (2013), the only difference being that they use
survivor functions $1-H\left(t\right)$ rather than the distribution
functions $H\left(t\right)$. We see that $c_{r}\left(t\right)$ and
$C\left(t\right)$ in the preceding analysis correspond to $c_{r}^{\left(2\right)}\left(t\right)$
and $C^{\left(2\right)}\left(t\right)$, respectively.

We also introduce \emph{n}-\emph{th order cumulative contrasts}. Conditioned
on $R=r$, we define

\begin{equation}
c_{r}^{\left[1\right]}\left(0,t\right)=c_{r}^{\left[1\right]}\left(t,\infty\right)=h_{1r}\left(t\right)-h_{2r}\left(t\right),
\end{equation}

\begin{equation}
c_{r}^{\left[2\right]}\left(0,t\right)=\int_{0}^{t}c_{r}^{\left(2\right)}\left(t_{1}\right)dt_{1},\quad c_{r}^{\left[2\right]}\left(t,\infty\right)=\int_{t}^{\infty}c_{r}^{\left(2\right)}\left(t_{1}\right)dt_{1},
\end{equation}

\begin{equation}
c_{r}^{\left[3\right]}\left(0,t\right)=\int_{0}^{t}\int_{0}^{t_{1}}c_{r}^{\left(3\right)}\left(t_{2}\right)dt_{2}dt_{1},\quad c_{r}^{\left[3\right]}\left(t,\infty\right)=\int_{t}^{\infty}\int_{t_{1}}^{\infty}c_{r}^{\left(3\right)}\left(t_{2}\right)dt_{2}dt_{1},
\end{equation}
etc. Generalizing,
\begin{equation}
c_{r}^{\left[n\right]}\left(0,t\right)=\int_{0}^{t}\left(\int_{0}^{t_{1}}\ldots\int_{0}^{t_{n-2}}c_{r}^{\left(n\right)}\left(t_{n-1}\right)dt_{n-1}\ldots dt_{2}\right)dt_{1},
\end{equation}

\begin{equation}
c_{r}^{\left[n\right]}\left(t,\infty\right)=\int_{t}^{\infty}\left(\int_{t_{1}}^{\infty}\ldots\int_{t_{n-2}}^{\infty}c_{r}^{\left(n\right)}\left(t_{n-1}\right)dt_{n-1}\ldots dt_{2}\right)dt_{1}.
\end{equation}
The corresponding unconditional cumulative contrasts of the \emph{n}-th
order are, as always, defined by integration of the conditional ones:
\begin{equation}
C^{\left[n\right]}\left(0,t\right)=\int_{\mathcal{R}}c_{r}^{\left[n\right]}\left(0,t\right)\dd\mu_{r}=\int_{0}^{t}\left(\int_{0}^{t_{1}}\ldots\int_{0}^{t_{n-2}}C^{\left(n\right)}\left(t_{n-1}\right)dt_{n-1}\ldots dt_{2}\right)dt_{1},
\end{equation}
\begin{equation}
C^{\left[n\right]}\left(t,\infty\right)=\int_{\mathcal{R}}c_{r}^{\left[n\right]}\left(t,\infty\right)\dd\mu_{r}=\int_{t}^{\infty}\left(\int_{t_{1}}^{\infty}\ldots\int_{t_{n-2}}^{\infty}C^{\left(n\right)}\left(t_{n-1}\right)dt_{n-1}\ldots dt_{2}\right)dt_{1}.
\end{equation}

In the proofs below we will make use of the recursive representation
of the conditional cumulative contrasts $c_{r}^{\left[n\right]}$.
It is verified by straightforward calculus. Denoting 
\begin{equation}
c_{i_{w}r}^{\left(n-1\right)}\left(t\right)=\underset{i_{1},...,i_{w-1},i_{w+1},...,i_{n}}{\sum}\left(-1\right)^{n-1-i_{w}+\sum_{k=1}^{n}i_{k}}h_{i_{1}...i_{w-1}i_{w}i_{w+1}...i_{n}r}\left(t\right),
\end{equation}
where $w\in\{1,\ldots,n\}$ and $i_{w}$ is fixed at 1 or 2, we have:

\begin{equation}
c_{r}^{\left[1\right]}\left(0,t\right)=c_{r}^{\left[1\right]}\left(t,\infty\right)=h_{1r}\left(t\right)-h_{2r}\left(t\right),
\end{equation}

\begin{align}
c_{r}^{\left[2\right]}\left(0,t\right) & =\int_{0}^{t}c_{r}^{\left(2\right)}\left(\tau\right)d\tau\nonumber \\
= & \int_{0}^{t}\left(h_{11r}\left(\tau\right)-h_{12r}\left(\tau\right)-h_{21r}\left(\tau\right)+h_{22r}\left(\tau\right)\right)d\tau\nonumber \\
= & \int_{0}^{t}\left[c_{i_{w}=1,r}^{(1)}\left(\tau\right)-c_{i_{w}=2,r}^{(1)}\left(\tau\right)\right]d\tau\\
= & \int_{0}^{t}c_{i_{w}=1,r}^{[1]}\left(0,\tau\right)d\tau-\int_{0}^{t}c_{i_{w}=2,r}^{[1]}\left(0,\tau\right)d\tau,\nonumber 
\end{align}
\begin{align}
c_{r}^{\left[2\right]}\left(t,\infty\right) & =\int_{t}^{\infty}c_{r}^{\left(2\right)}\left(\tau\right)d\tau\nonumber \\
= & \int_{t}^{\infty}\left(h_{11r}\left(\tau\right)-h_{12r}\left(\tau\right)-h_{21r}\left(\tau\right)+h_{22r}\left(\tau\right)\right)d\tau\\
= & \int_{t}^{\infty}\left[c_{i_{w}=1,r}^{(1)}\left(\tau\right)-c_{i_{w}=2,r}^{(1)}\left(\tau\right)\right]d\tau\nonumber \\
= & \int_{t}^{\infty}c_{i_{w}=1,r}^{[1]}\left(\tau,\infty\right)d\tau-\int_{t}^{\infty}c_{i_{w}=2,r}^{[1]}\left(\tau,\infty\right)d\tau,\nonumber 
\end{align}
\begin{align}
c_{r}^{\left[3\right]}\left(0,t\right) & =\int_{0}^{t}\int_{0}^{t_{1}}c_{r}^{\left(3\right)}\left(t_{2}\right)dt_{2}dt_{1}\nonumber \\
= & \int_{0}^{t}\int_{0}^{t_{1}}\left[c_{i_{w}=1,r}^{\left(2\right)}\left(t_{2}\right)-c_{i_{w}=2,r}^{\left(2\right)}\left(t_{2}\right)\right]dt_{2}dt_{1}\nonumber \\
= & \int_{0}^{t}\left[\int_{0}^{t_{1}}c_{i_{w}=1,r}^{\left(2\right)}\left(t_{2}\right)dt_{2}-\int_{0}^{t_{1}}c_{i_{w}=2,r}^{\left(2\right)}\left(t_{2}\right)dt_{2}\right]dt_{1}\\
= & \int_{0}^{t}c_{i_{w}=1,r}^{\left[2\right]}\left(0,\tau\right)d\tau-\int_{0}^{t}c_{i_{w}=2,r}^{\left[2\right]}\left(0,\tau\right)d\tau,\nonumber 
\end{align}
\begin{align}
c_{r}^{\left[3\right]}\left(t,\infty\right) & =\int_{t}^{\infty}\int_{t_{1}}^{\infty}c_{r}^{\left(3\right)}\left(t_{2}\right)dt_{2}dt_{1}\nonumber \\
= & \int_{t}^{\infty}\int_{t_{1}}^{\infty}\left[c_{i_{w}=1,r}^{\left(2\right)}\left(t_{2}\right)-c_{i_{w}=2,r}^{\left(2\right)}\left(t_{2}\right)\right]dt_{2}dt_{1}\nonumber \\
= & \int_{t}^{\infty}\left[\int_{t_{1}}^{\infty}c_{i_{w}=1,r}^{\left(2\right)}\left(t_{2}\right)dt_{2}-\int_{t_{1}}^{\infty}c_{i_{w}=2,r}^{\left(2\right)}\left(t_{2}\right)dt_{2}\right]dt_{1}\\
= & \int_{t}^{\infty}c_{i_{w}=1,r}^{\left[2\right]}\left(\tau,\infty\right)d\tau-\int_{t}^{\infty}c_{i_{w}=2,r}^{\left[2\right]}\left(\tau,\infty\right)d\tau,\nonumber 
\end{align}
and generally, for $n>1$,

\begin{equation}
c_{r}^{\left[n\right]}\left(0,t\right)=\int_{0}^{t}c_{i_{w}=1,r}^{\left[n-1\right]}\left(0,\tau\right)d\tau-\int_{0}^{t}c_{i_{w}=2,r}^{\left[n-1\right]}\left(0,\tau\right)d\tau,\label{eq:recursive1}
\end{equation}
\begin{equation}
c^{\left[n\right]}\left(t,\infty\right)=\int_{t}^{\infty}c_{i_{w}=1,r}^{\left[n-1\right]}\left(\tau,\infty\right)d\tau-\int_{t}^{\infty}c_{i_{w}=2,r}^{\left[n-1\right]}\left(\tau,\infty\right)d\tau.\label{eq:recursive2}
\end{equation}

Also we have, by substitution of variables under integral,

\begin{equation}
c_{i_{w}r}^{\left[n-1\right]}\left(0,t\right)=c_{r}^{\left[n-1\right]}\left(0,t-x_{i_{w}r}^{w}\right),
\end{equation}

\begin{equation}
c_{i_{wr}}^{\left[n-1\right]}\left(t,\infty\right)=c^{\left[n-1\right]}\left(t-x_{i_{w}r}^{w},\infty\right).
\end{equation}

The Prolongation Assumption generalizing (\ref{eq:prolongation 2x2})-(\ref{eq:default 2x2})
is formulated as follows. \medskip{}

\noindent\textbf{Prolongation Assumption.} $R$ and functions $x^{1},\ldots,x^{n}$
in (\ref{eq:function for n}) can be chosen so that $x_{1r}^{k}\leq x_{2r}^{k}$
for all $R=r$ and for all $k=1,\ldots,n$. Without loss of generality,
we can also assume $x_{1r}^{1}\leq x_{1r}^{2}\leq\ldots\leq x_{1r}^{n}$
(if not, rearrange $x_{1r}^{1},\ldots,x_{1r}^{n}$). 

\medskip{}

\noindent\textbf{Notation Convention.} As we did before for $n=2$,
when $r$ is fixed throughout a discussion, we omit this argument
and write $x_{i_{1}}^{1},\ldots,x_{i_{n}}^{n}$, $t_{i_{1}i_{2}...i_{n}}$,
$h_{i_{1}i_{2}...i_{n}}\left(t\right)$, $c^{(n)}(t)$ in place of
$x_{i_{1}r}^{1},\ldots,x_{i_{n}r}^{n}$, $t_{i_{1}i_{2}...i_{n}r}$,
$h_{i_{1}i_{2}...i_{n}r}\left(t\right)$, $c_{r}^{(n)}(t)$.

\subsection{Parallel Processes}

\subsubsection{Simple Parallel Architectures of Size $n$}
\begin{thm}
\label{thm:OR_n_simple}If $T=X^{1}\wedge\ldots\wedge X^{n}$, then
for any $r,t$, $c^{\left(n\right)}\left(t\right)\leq0$ if $n$ is
even and $c^{\left(n\right)}\left(t\right)\geq0$ if $n$ is odd.
Consequently, for any $t$, $C^{\left(n\right)}\left(t\right)\leq0$
if $n$ is even and $C^{\left(n\right)}\left(t\right)\geq0$ if $n$
is odd.\end{thm}
\begin{proof}
By induction on $n$, the case $n=1$ being true by the Prolongation
Assumption: 
\[
c^{\left(1\right)}\left(t\right)=h_{1}\left(t\right)-h_{2}\left(t\right)\geq0.
\]
Let the statement of the theorem be true for $c^{(n-1)}(t)$ , with
$n-1\geq1$. By the Prolongation Assumption, 
\[
t_{1i_{2}...i_{n}}=x_{1}^{1}\wedge x_{i_{2}}^{2}\wedge\ldots\wedge x_{i_{n}}^{n}=x_{1}^{1},
\]
for any $i_{2}\ldots i_{n}$, whence

\[
h_{1i_{2}...i_{n}}\left(t\right)=\left\{ \begin{array}{c}
0,\:if\:t<x_{1}^{1}\\
1,\:if\:t\geq x_{1}^{1}
\end{array}.\right.
\]
Therefore $c_{i_{1}=1}^{\left(n-1\right)}\left(t\right)=0$, and,
applying the induction hypothesis to $c_{i_{1}=2}^{\left(n-1\right)}\left(t\right)$,

\begin{eqnarray*}
c^{\left(n\right)}\left(t\right) & = & c_{i_{1}=1}^{\left(n-1\right)}\left(t\right)-c_{i_{1}=2}^{\left(n-1\right)}\left(t\right)=-c_{i_{1}=2}^{\left(n-1\right)}\left(t\right)=\left\{ \begin{array}{c}
\leq0,\textnormal{ if }n\textnormal{ is even}\\
\geq0,\textnormal{ if }n\textnormal{ is odd}
\end{array}\right..
\end{eqnarray*}
That $C^{\left(n\right)}\left(t\right)\leq0$ if $n$ is even and
$C^{\left(n\right)}\left(t\right)\geq0$ if $n$ is odd follows by
the standard argument. \end{proof}
\begin{thm}
\label{thm:AND_n_simple}If $T=X^{1}\vee\ldots\vee X^{n}$, then for
any $r,t$, $c^{\left(n\right)}\left(t\right)\geq0$. Consequently,
for any $t$, $C^{\left(n\right)}\left(t\right)\geq0$. \end{thm}
\begin{proof}
By induction on $n$, the case $n=1$ being true by the Prolongation
Assumption: 
\[
c^{\left(1\right)}\left(t\right)=h_{1}\left(t\right)-h_{2}\left(t\right)\geq0.
\]
Let the theorem be true for $c^{(n-1)}(t)$, where $n-1\geq1$. Let
\[
x_{2}^{1}\vee x_{2}^{2}\vee\ldots\vee x_{2}^{n}=x_{2}^{m},
\]
where $1\leq m\leq n$. We have then 
\[
t_{i_{1}i_{2}...i_{m-1}2i_{m+1}...i_{n}}=x_{2}^{m},
\]
 and 
\[
h_{i_{1}...i_{m-1}2i_{m+1}...i_{n}}(t)=\left\{ \begin{array}{c}
0,\:if\:t<x_{2}^{m}\\
1,\:if\:t\geq x_{2}^{m}
\end{array},\right.
\]
for all $i_{1}...i_{m-1},i_{m+1}...i_{n}$. Then $c_{i_{m}=2}^{\left(n-1\right)}\left(t\right)=0$,
and
\begin{eqnarray*}
c^{\left(n\right)}\left(t\right) & = & c_{i_{m}=1}^{\left(n-1\right)}\left(t\right)-c_{i_{m}=2}^{\left(n-1\right)}\left(t\right)=c_{i_{m}=1}^{\left(n-1\right)}\left(t\right)\geq0.
\end{eqnarray*}
Consequently, $C^{\left(n\right)}\left(t\right)\geq0$, for any $t$.
\end{proof}

\subsubsection{Multiple Parallel Processes in Arbitrary $\SP$ Networks}

In a composition $\SP\left(X^{1},\ldots,X^{n},\ldots\right)$, the
components $X^{1},\ldots,X^{n}$ are considered parallel if any two
of them are parallel. We assume selective influences $(X^{1},\dots,X^{n},\ldots)\looparrowleft(\lambda^{1},\ldots,\lambda^{n},\emptyset)$.
We do not consider the complex situation when some of the selectively
influenced processes $X^{1},\ldots,X^{n}$ are min-parallel and some
are max-parallel. However, if they are all (pairwise) min-parallel
or all max-parallel, we have essentially the same situation as with
a simple parallel arrangement of $n$ durations.
\begin{lem}
\label{lem:parallel decomposition n}If $X^{1},\ldots,X^{n}$ are
all min-parallel or max-parallel in an $\SP$ composition, this composition
can be presented as $T=A^{1}\wedge\ldots\wedge A^{n}$ or $T=A^{1}\vee\ldots\vee A^{n}$,
respectively. In either case, $(A^{1},\dots,A^{n})\looparrowleft(\lambda^{1},\ldots,\lambda^{n})$
and the Prolongation Assumption holds for any $R=r$.\end{lem}
\begin{proof}
For the min-parallel case, by a minor modification of the proof of
Lemma \ref{lem:parallel decomposition 2} we present the $\SP$ composition
as 

\[
\underset{=A^{1}}{\underbrace{\SP^{1}(X^{1},\ldots)}}\wedge\SP^{2}(X^{2},\ldots,X^{n},\ldots),
\]
or
\[
\left(\SP^{1}(X^{1},\ldots)\wedge\SP^{2}(X^{2},\ldots,X^{n},\ldots)+X\right)\wedge Y=\underset{=A^{1}}{\underbrace{\left(\SP^{1}(X^{1},\ldots)+X\right)}}\wedge\left(\SP^{2}(X^{2},\ldots,X^{n},\ldots)+X\right)\wedge Y,
\]
or else
\[
\left(\SP^{1}(X^{1},\ldots)\wedge\SP^{2}(X^{2},\ldots,X^{n},\ldots)\wedge X\right)+Y=\underset{=A^{1}}{\underbrace{\left(\SP^{1}(X^{1},\ldots)+Y\right)}}\wedge\left(\SP^{2}(X^{2},\ldots,X^{n},\ldots)\wedge X+Y\right).
\]

Then we analogously decompose $\SP^{2}(X^{2},\ldots,X^{n},\ldots)$
achieving $A^{1}\wedge A^{2}\wedge\SP^{3}(X^{3},\ldots,X^{n},\ldots)$,
and proceed in this fashion until we reach the required $A^{1}\wedge\ldots\wedge A^{n}$.
The pattern of selective influences is seen immediately, and the Prolongation
Assumption follows by the monotonicity of the $\SP$ compositions.
The proof for the max-parallel case is analogous.\end{proof}
\begin{thm}
\label{thm:paralell general n}If $X^{1},\ldots,X^{n}$ are min-parallel
in an $\SP$ composition, then for any $r,t$, $c^{\left(n\right)}\left(t\right)\leq0$
if $n$ is even and $c^{\left(n\right)}\left(t\right)\geq0$ if $n$
is odd. Consequently, for any $t$, $C^{\left(n\right)}\left(t\right)\leq0$
if $n$ is even and $C^{\left(n\right)}\left(t\right)\geq0$ if $n$
is odd. If $X^{1},\ldots,X^{n}$ are max-parallel, then for any $r,t$,
$c^{\left(n\right)}\left(t\right)\geq0$, and for any $t$, $C^{\left(n\right)}\left(t\right)\geq0$. \end{thm}
\begin{proof}
Follows from Lemma \ref{lem:parallel decomposition n} and Theorems
\ref{thm:OR_n_simple} and \ref{thm:AND_n_simple}.
\end{proof}

\subsection{Sequential Processes }

\subsubsection{Simple Serial Architectures of Size $n$}
\begin{thm}
\label{thm:PLUS_n_simple}If $T=X^{1}+\ldots+X^{n}$, then for any
$r,t$, $c^{\left[n\right]}\left(0,t\right)\geq0$, while $c^{\left[n\right]}\left(t,\infty\right)\leq0$
if n is even and $c^{\left[n\right]}\left(t,\infty\right)\geq0$ if
n is odd; moreover, $c^{\left[n\right]}(0,\infty)=0$ for any $r$.
Consequently, for any $t$, $C^{\left[n\right]}\left(0,t\right)\geq0$,
while $C^{\left[n\right]}\left(t,\infty\right)\leq0$ if n is even
and $C^{\left[n\right]}\left(t,\infty\right)\geq0$ if n is odd; moreover,
$C^{\left[n\right]}\left(0,\infty\right)=0$.\end{thm}
\begin{proof}
By induction on $n$, the case $n=1$ being true by the Prolongation
Assumption: 
\[
c^{\left[1\right]}\left(0,t\right)=c^{\left[1\right]}\left(t,\infty\right)=h_{1}\left(t\right)-h_{2}\left(t\right)\geq0,
\]
and
\[
\lim_{t\rightarrow\infty}c^{\left[1\right]}\left(0,t\right)=\lim_{t\rightarrow0}c^{\left[1\right]}\left(t,\infty\right)=0.
\]
Let the statement of the theorem hold for all natural numbers up to
and including $n-1\geq1$. Using the recursive representations (\ref{eq:recursive1})-(\ref{eq:recursive2}),
\begin{equation}
\begin{array}{l}
c^{\left[n\right]}\left(0,t\right)=\int_{0}^{t}c_{i_{w}=1}^{\left[n-1\right]}\left(0,\tau\right)d\tau-\int_{0}^{t}c_{i_{w}=2}^{\left[n-1\right]}\left(0,\tau\right)d\tau\\
=\int_{0}^{t-x_{1}^{w}}c^{\left[n-1\right]}\left(0,\tau\right)d\tau-\int_{0}^{t-x_{2}^{w}}c^{\left[n-1\right]}\left(0,\tau\right)d\tau\\
=\int_{t-x_{2}^{w}}^{t-x_{1}^{w}}c^{\left[n-1\right]}\left(0,\tau\right)d\tau
\end{array},\label{eq:4}
\end{equation}
which is $\geq0$ since $c^{\left[n-1\right]}\left(0,\tau\right)\geq0$
and $t-x_{2}^{w}\leq t-x_{1}^{w}$. Analogously,

\begin{equation}
\begin{array}{l}
c^{\left[n\right]}\left(t,\infty\right)=\int_{t}^{\infty}c_{i_{w}=1}^{\left[n-1\right]}\left(\tau,\infty\right)d\tau-\int_{t}^{\infty}c_{i_{w}=2}^{\left[n-1\right]}\left(\tau,\infty\right)d\tau\\
=\int_{t-x_{1}^{w}}^{\infty}c^{\left[n-1\right]}\left(\tau,\infty\right)d\tau-\int_{t-x_{2}^{w}}^{\infty}c^{\left[n-1\right]}\left(\tau,\infty\right)d\tau\\
=-\int_{t-x_{2}^{w}}^{t-x_{1}^{w}}c^{\left[n-1\right]}\left(\tau,\infty\right)d\tau
\end{array},\label{eq:5}
\end{equation}
which is $\leq0$ if $n$ is even and $\geq0$ if $n$ is odd. Applying
the mean value theorem to the results of (\ref{eq:4}) and (\ref{eq:5}),
we get, for some $t-x_{2}^{w}<t',t''<t-x_{1}^{w}$ 
\[
\int_{t-x_{2}^{w}}^{t-x_{1}^{w}}c^{\left[n-1\right]}\left(0,\tau\right)d\tau=c^{\left[n-1\right]}\left(0,t'\right)\left(-x_{1}^{w}+x_{2}^{w}\right),
\]
\[
\int_{t-x_{2}^{w}}^{t-x_{1}^{w}}c^{\left[n-1\right]}\left(\tau,\infty\right)d\tau=c^{\left[n-1\right]}\left(t'',\infty\right)\left(-x_{1}^{w}+x_{2}^{w}\right),
\]
and, as $c^{\left[n-1\right]}\left(0,\infty\right)=0$, both expressions
tend to zero as, respectively, $t\rightarrow\infty$ (implying $t'\rightarrow\infty$)
and $t\rightarrow0$ (implying $t''\rightarrow0$).
\end{proof}

\subsubsection{Multiple Sequential Processes in Arbitrary $\SP$ Networks}

In a composition $\SP\left(X^{1},\ldots,X^{n},\ldots\right)$, the
components $X^{1},\ldots,X^{n}$ are considered sequential if any
two of them are sequential. By analogy with Theorem \ref{thm:2x2 general-1-1}
for two sequential processes and with Theorem \ref{thm:paralell general n}
for parallel $X^{1},\ldots,X^{n}$, one might expect that the result
for the simple sequential arrangement $X^{1}+\ldots+X^{n}$ will also
extend to $n$ sequential components of more complex compositions
$\SP\left(X^{1},\ldots,X^{n},\ldots\right)$. However, this is not
the case, as one can see from the following counterexample.

Consider the composition 
\begin{equation}
\SP(X^{1},X^{2},X^{3},Y)=\left(X^{1}+X^{2}+X^{3}\right)\wedge\left(Y=2\right),
\end{equation}
with $\left(X^{1},X^{2},X^{3}\right)$ selectively influenced by binary
factors, so that
\begin{equation}
\begin{array}{l}
x_{1}^{1}=x_{1}^{2}=x_{1}^{3}=0,\\
x_{2}^{1}=x_{2}^{2}=x_{2}^{3}=1.
\end{array}
\end{equation}
It follows that
\begin{equation}
\begin{array}{l}
t_{111}=0,\\
t_{112}=t_{121}=t_{211}=1,\\
t_{122}=t_{212}=t_{221}=t_{222}=2.
\end{array}
\end{equation}
This is clearly a sequential arrangement of the three durations $X^{1},X^{2},X^{3}$,
but one can easily check that $c^{\left[3\right]}\left(0,t\right)$
here is not nonnegative for all $t$. For instance, at $\mbox{t=3}$
we have, by direct computation, $c^{\left[3\right]}\left(0,t\right)=-1$.
We conclude that there is no straightforward generalization of Theorems
\ref{thm:PLUS_n_simple} to arbitrary $\SP$ compositions.

\section{Conclusion}

The work presented in this paper is summarized in the abstract. By
proving and generalizing most of the known results on the interaction
contrast of distribution functions, we have demonstrated a new way
of dealing with $\SP$ mental architectures. It is based on conditioning
all hypothetical components of a network on a fixed value of a common
source of randomness $R$ (the ``hidden variable'' of the contextuality
analysis in quantum theory), which renders these components deterministic
quantities, and then treating these deterministic quantities as random
variables with shifted Heaviside distribution functions. 

The potential advantage of this method can be seen in the fact that
the shifted Heaviside functions have the simplest possible arithmetic
among distribution functions: for every time moment it only involves
0's and 1's. As a result, the complexity of this arithmetic does not
increase with nonlinearity of the relations involved. Thus, Dzhafarov
and Schweickert (1995), Cortese and Dzhafarov (1996), and Dzhafarov
and Cortese (1996) argued that composition rules for mental architectures
need not be confined to $+,\wedge,\vee$. They analyzed architectures
involving other associative and commutative operations, such as multiplication.
Due to mathematical complexity, however, this work was confined to
networks consisting of two components that are either stochastically
independent or monotone functions of each other. It remains to be
seen whether the approach presented here, \emph{mutatis mutandis},
will lead to significant generalizations in this line of work. 

The limitations of the approach, however, are already apparent. Thus,
we were not able to achieve any progress over known results in applying
it to Wheatstone bridges (Schweickert \& Giorgini, 1999; Dzhafarov
et al., 2004). The possibility that the ``architecture'' (composition
rule) itself changes as one changes experimental factors makes the
perspective of a general theory based on our approach even more problematic
(e.g., Townsend \& Fific, 2004). It seems, however, that these problems
are not specific for just our approach.

\subsection*{Acknowledgments}

This work is supported by NSF grant SES-1155956 and AFOSR grant FA9550-14-1-0318.

\section*{REFERENCES}

\setlength{\parindent}{0cm}\everypar={\hangindent=15pt}

Bell, J. (1964). On the Einstein-Podolsky-Rosen paradox. \emph{Physics,}
1, 195-200. 

Bell, J. (1966). On the problem of hidden variables in quantum mechanics.
Review of Modern Physics, 38, 447-453.

Cereceda, J. (2000). Quantum mechanical probabilities and general
probabilistic constraints for Einstein--Podolsky--Rosen--Bohm experiments.
\textit{Foundations of Physics Letters}\textit{\emph{,}} 13, 427--442.

Cortese, J. M., \& Dzhafarov, E. N. (1996). Empirical recovery of
response time decomposition rules II: Discriminability of serial and
parallel architectures. \emph{Journal of Mathematical Psychology},
40, 203-218.

Dzhafarov, E. N. (2003). Selective influence through conditional independence.
\emph{Psychometrika}, 68, 7\textendash 26.

Dzhafarov, E. N., \& Cortese, J. M. (1996). Empirical recovery of
response time decomposition rules I: Sample-level Decomposition tests.
\emph{Journal of Mathematical Psychology}, 40, 185-202.

Dzhafarov, E. N., \& Gluhovsky, I. (2006). Notes on selective influence,
probabilistic causality, and probabilistic dimensionality. \emph{Journal
of Mathematical Psychology}, 50, 390-401.

Dzhafarov, E. N., \& Kujala, J. V. (2010). The Joint Distribution
Criterion and the Distance Tests for selective probabilistic causality.
\emph{Frontiers in Psychology} 1:151 doi: 10.3389/fpsyg.2010.00151.

Dzhafarov, E. N., \& Kujala, J. V. (2012a). Selectivity in probabilistic
causality: Where psychology runs into quantum physics. \emph{Journal
of Mathematical Psychology}, 56, 54-63.

Dzhafarov, E. N., \& Kujala, J. V. (2012b). Quantum entanglement and
the issue of selective influences in psychology: An overview. \emph{Lecture
Notes in Computer Science} 7620, 184-195.

Dzhafarov, E. N., \& Kujala, J. V. (2013). Order-distance and other
metric-like functions on jointly distributed random variables. \emph{Proceedings
of the American Mathematical Society,} 141, 3291-3301.

Dzhafarov, E. N., \& Kujala, J. V. (2014a). On selective influences,
marginal selectivity, and Bell/CHSH inequalities. \emph{Topics in
Cognitive Science} 6, 121-128.

Dzhafarov, E. N., \& Kujala, J. V. (2014b). Embedding quantum into
classical: contextualization vs conditionalization. PLoS One 9(3):e92818.
doi:10.1371/journal.pone.0092818.

Dzhafarov, E. N., \& Kujala, J. V. (2014c). A qualified Kolmogorovian
account of probabilistic contextuality. \emph{Lecture Notes in Computer
Science} 8369, 201--212.

Dzhafarov, E. N., \& Kujala, J. V. (2014d). Contextuality is about
identity of random variables. \emph{Physica Scripta} T163, 014009.

Dzhafarov, E.N., \& Kujala, J.V. (in press). Probability, random variables,
and selectivity. In W.Batchelder, H. Colonius, E.N. Dzhafarov, J.
Myung (Eds). \emph{The New Handbook of Mathematical Psychology}. Cambridge,
UK: Cambridge University Press.

Dzhafarov, E. N., Kujala, J. V., \& Larsson, J.-Å. (2015). Contextuality
in three types of quantum-mechanical systems. \emph{Foundations of
Physics} 2015, doi: 10.1007/s10701-015-9882-9.

Dzhafarov, E. N., \& Schweickert, R. (1995). Decompositions of response
times: An almost general theory. \emph{Journal of Mathematical Psychology},
39, 285-314.

Dzhafarov, E. N., Schweickert, R., \& Sung, K. (2004). Mental architectures
with selectively influenced but stochastically interdependent components.
\emph{Journal of Mathematical Psychology}, 48, 51-64.

Khrennikov, A. Yu. (2009). \emph{Contextual Approach to Quantum Formalism}.
New York: Springer.

Kochen. S., \& Specker, E. P. (1967). The problem of hidden variables
in quantum mechanics. \emph{Journal of Mathematics and Mechanics},
17, 59\textendash 87.

Kujala, J. V., \& Dzhafarov, E. N. (2008). Testing for selectivity
in the dependence of random variables on external factors. \emph{Journal
of Mathematical Psychology}, 52, 128\textendash 144.

Landau, L. J. (1988). Empirical two-point correlation functions. \emph{Foundations
of Physics}, 18, 449\textendash 460.

Masanes, Ll. , Acin, A. , \& Gisin, N. (2006). General properties
of nonsignaling theories. \textit{Physical Review} A 73, 012112.

Popescu, S. \& Rohrlich, D. (1994). Quantum nonlocality as an axiom.
\emph{Foundations of Physics}, 24, 379--385.

Roberts, S., \& Sternberg, S. (1993). The meaning of additive reaction-
time effects: Tests of three alternatives. In D. E. Meyer \& S. Kornblum
(Eds.), \emph{Attention and performance XIV: Synergies in experimental
psychology, artificial intelligence, and cognitive neuroscience} (pp.
611\textendash 654). Cambridge, MA: MIT Press.

Schweickert, R., Fisher, D.L., \& Sung, K. (2012). Discovering Cognitive
Architecture by Selectively Influencing Mental Processes. New Jersey:
World Scientific.

Schweickert, R., \& Giorgini, M. (1999). Response time distributions:
Some simple effects of factors selectively influencing mental processes.
\emph{Psychonomic Bulletin \& Review}, 6, 269\textendash 288.

Schweickert, R., Giorgini, M., \& Dzhafarov, E. N. (2000). Selective
influence and response time cumulative distribution functions in serial-parallel
task networks. \emph{Journal of Mathematical Psychology}, 44, 504-535.

Schweickert, R., \& Townsend, J. T. (1989). A trichotomy: Interactions
of factors prolonging sequential and concurrent mental processes in
stochastic discrete mental (PERT) networks. \emph{Journal of Mathematical
Psychology}, 33, 328\textendash 347.

Sternberg, S. (1969). The discovery of processing stages: Extensions
of Donders\textquoteright{} method. In W.G. Koster (Ed.), \emph{Attention
and Performance II. Acta Psychologica}, 30, 276\textendash 315.

Suppes, P., \& Zanotti, M. (1981). When are probabilistic explanations
possible? Synthese 48:191--199.

Townsend, J. T. (1984). Uncovering mental processes with factorial
experiments. \emph{Journal of Mathematical Psychology}, 28, 363\textendash 400. 

Townsend, J. T. (1990a). The truth and consequences of ordinal differences
in statistical distributions: Toward a theory of hierarchical inference.
\emph{Psychological Bulletin}, 108, 551-567. 

Townsend, J. T. (1990b). Serial vs. parallel processing: Sometimes
they look like tweedledum and tweedledee but they can (and should)
be distinguished. \emph{Psychological Science}, 1, 46-54.

Townsend, J. T., \& Fific, M. (2004). Parallel \& serial processing
and individual differences in high-speed scanning in human memory.
\emph{Perception \& Psychophysics}, 66, 953-962.

Townsend, J. T., \& Nozawa, G. (1995). Spatio-temporal properties
of elementary perception: An investigation of parallel, serial, and
coactive theories. \emph{Journal of Mathematical Psychology}, 39,
321\textendash 359.

Townsend, J. T., \& Schweickert, R. (1989). Toward the trichotomy
method of reaction times: Laying the foundation of stochastic mental
networks. \emph{Journal of Mathematical Psychology}, 33, 309\textendash 327.

Yang, H., Fific, M., \& Townsend, J. T. (2013). Survivor interaction
contrast wiggle predictions of parallel and serial models for an arbitrary
number of processes. \emph{Journal of Mathematical Psychology}, 58,
21-32. 
\end{document}